\documentclass[10pt,aps,prb,twocolumn, nofootinbib]{revtex4-1}

\usepackage{amssymb}
\usepackage{amsmath}
\usepackage{amsthm}
\usepackage{subfig}
\usepackage{graphicx}
\usepackage{color}
\usepackage{hyperref}
\newtheorem{theorem}{Theorem}
\usepackage{listings}
\usepackage{tkz-euclide}
\usepackage{textcomp}
\hypersetup{        
    unicode=false,          
    pdftoolbar=true,        
    pdfmenubar=true,        
    pdffitwindow=false,     
    pdfstartview={FitH},    
    pdfauthor={Christopher Abel},     
    colorlinks=true,       
    linkcolor=blue,          
    citecolor=red,        
}
\begin{document}



\title{Deriving Kepler\textquotesingle s Laws Using Quaternions}

\author{Christopher J. Abel}



\begin{abstract}

In the past, Kepler painstakingly derived laws of planetary motion using difficult to understand and hard to follow techniques. In 1843 William Hamilton created and described the quaternions, which extend the complex numbers and can easily describe rotations in three dimensional space. In this article, we will harness this system to provide a new and intuitive way to derive Kepler's laws. This will include using a quaternionic version of the spatial Kepler problem differential equation, and using the general solution to describe the motion of planets orbiting a central body. We use the standard method for regularizing celestial mechanics, but this article will be solely focused on showing the validity of Kepler's laws.
\end{abstract}



\maketitle

\section{The Star Conjugate}

In order to describe the KS transformation in the form of quaternions we require a new operation that will help remove the fourth dimension for any general quaternion. This setup was suggested by \citet{Waldvogel2006}. We define the \textit{star conjugate} of a quaternion $\boldsymbol{q}\in\mathbb{H}$ to be
\begin{equation}\label{equation:3.1}
\boldsymbol{q}^*=q_0+q_1i+q_2j-q_3k.
\end{equation}
Note that all the star conjugate does is negate the $k$ term. We can easily represent this in terms of the quaternion conjugate,
\begin{equation}\label{equation:3.2}
\boldsymbol{q}^*=(-k)\overline{\boldsymbol{q}}k.
\end{equation}
One may easily show that for any $\boldsymbol{q},\boldsymbol{v}\in\mathbb{H}$,
\begin{equation}\label{equation:3.3}
(\boldsymbol{q}^*)^*=\boldsymbol{q},\qquad |\boldsymbol{q}^*|=|\boldsymbol{q}|, \qquad (\boldsymbol{q}\boldsymbol{v})^*=\boldsymbol{v}^*\boldsymbol{q}^*.
\end{equation}

\medskip
We now define a mapping between the quaternion algebra to take advantage of the star conjugate. Consider the mapping,
\begin{equation}\label{equation:3.4}
\boldsymbol{q}\mapsto\boldsymbol{x}=\boldsymbol{q}\boldsymbol{q}^*.
\end{equation}
Immediately we notice,
\begin{equation}\label{equation:3.5}
\boldsymbol{x}^*=(\boldsymbol{q}\boldsymbol{q}^*)^*=(\boldsymbol{q}^*)^*\boldsymbol{q}^*=\boldsymbol{q}\boldsymbol{q}^*=\boldsymbol{x},
\end{equation}
which implies $\boldsymbol{x}$ has the form
\begin{equation}\label{equation:3.6}
\boldsymbol{x}=x_0+x_1i+x_2j.
\end{equation}
To find the values of $\{x_i\}_0^2$, we do the quaternion multiplication. The result yields
\begin{eqnarray}
x_0&=q_0^2-q_1^2-q_2^2+q_3^2\nonumber\\
x_1&=2(q_0q_1-q_2q_3)\label{equation:3.7}\\
x_2&=2(q_0q_2+q_1q_3)\nonumber.
\end{eqnarray}
This is the KS transformation in its classical form. Thus the KS transformation $q=(q_0,q_1,q_2,q_3)\in\mathbb{R}^4\mapsto x=(x_0,x_1,x_2)\in\mathbb{R}^3$ is given by the quaternion relation
\begin{equation}\label{equation:3.8}
\boldsymbol{x}=\boldsymbol{q}\boldsymbol{q}^*,
\end{equation}
where $\boldsymbol{x}=x_0+x_1i+x_2j$, $\boldsymbol{q}=q_0+q_1i+q_2j+q_3k$, and $\boldsymbol{q}^*$ is defined as (\ref{equation:3.1}).

\section{Differentiation}

Since the mapping (\ref{equation:3.4}) involving the star conjugate maps from four dimensions to three, we are left with one degree of freedom. We take advantage of this by simplifying differentiation. To do this we impose a bilinear relation between a quaternion $\boldsymbol{q}$ and its differential, $d\boldsymbol{q}$, being
\begin{equation}\label{equation:4.1}
\boldsymbol{q}\, d\boldsymbol{q}^*=d\boldsymbol{q}\, \boldsymbol{q}^*.
\end{equation}
Note (\ref{equation:4.1}) has the form of a commutator. The differential of the star conjugate mapping (\ref{equation:3.4}) then becomes
\begin{equation}\label{equation:4.2}
d\boldsymbol{x}=d\boldsymbol{q}\, \boldsymbol{q}^*+\boldsymbol{q}\, d\boldsymbol{q}^*=2d\boldsymbol{q}\, \boldsymbol{q}^*.
\end{equation}
Using these techniques, \citet{Waldvogel2008} reduces the spatial Kepler problem to
\begin{equation}\label{equation:5.16}
2\boldsymbol{u}''+h\boldsymbol{u}=0,
\end{equation}
with the parameter
\begin{equation}\label{equation:5.5}
dt=r\cdot d\tau, \quad \frac{d}{d\tau}()=()',
\end{equation}
where $t$ is time, $r$ is the radius, and $u=q^*$. For later importance, we note the energy integral:
\begin{equation}\label{equation:5.13}
-2\left|\boldsymbol{q}'\right|^2=-\mu+rh.
\end{equation}

\section{Deriving Kepler's laws}
In this section, we will derive Kepler's first law in the standard method for regularizing celestial mechanics, then after which we will prove Kepler's second and third laws using new techniques based upon differentiation and integration.
\subsection{Kepler's first law}

We first solve (\ref{equation:5.16}) to find the general solution,
\begin{equation}\label{equation:6.1}
\boldsymbol{u}=A\cos(\omega\tau)+B\sin(\omega\tau),\ \omega=\sqrt{\frac{h}{2}}, \ A,B\in\mathbb{H}.
\end{equation}
There are two equations we note before continuing. The first is the $\cos$ and $\sin$ functions are linear in $\mathbb{H}$, and thus are two dimensional. The second is since the equation for motion of $\boldsymbol{u}$ is two dimensional, we may rotate it into the complex plane without loss of generality, to make solving the fine details a bit easier. So after rotating to the complex plane, we rotate once more, such that the apocenter is a real number greater than zero. 

\medskip
Note that we are not scaling, translating, or making any changes to the equation at all, but merely changing its direction. Once we have solutions to the equation of motion in this form, we may rotate it back to the original equation of motion as needed.

\medskip
After doing our rotations, (\ref{equation:6.1}) becomes,
\begin{equation}\label{equation:6.2}
\boldsymbol{u}=A\cos(\omega\tau)+iB\sin(\omega\tau), \ A,B\in\mathbb{R}.
\end{equation}
Earlier we made the substitution $\boldsymbol{x}=\boldsymbol{q}\boldsymbol{q}^*=\boldsymbol{u}^*\boldsymbol{u}$, but with our rotations we have $\boldsymbol{u}\in\mathbb{C}$, not $\boldsymbol{u}\in\mathbb{H}$. Thus for our substitution, we need not include the star conjugate, so we have $\boldsymbol{x}=\boldsymbol{u}^2$. Therefore our equation for motion of $\boldsymbol{x}$ is,
\begin{align*}
\boldsymbol{x}&=\left(A\cos(\omega\tau)+iB\sin(\omega\tau)\right)^2\\
&=A^2\cos^2(\omega\tau)-B^2\sin^2(\omega\tau)+2AB\cos(\omega\tau)\sin(\omega\tau)
\end{align*}
\begin{equation}\label{equation:6.3}
=\frac{A^2-B^2}{2}+\frac{A^2+B^2}{2}\cos(2\omega\tau)+iAB\sin(2\omega\tau).
\end{equation}
Since $2\omega\tau$ defines the angle of the planet from the center, it must be the eccentric anomaly. Thus the eccentric anomaly, which we denote $E$, must be, 
\begin{equation}\label{equation:6.4}
E=2\omega\tau=\sqrt{2h}\tau.
\end{equation}
We substitute (\ref{equation:6.4}) into (\ref{equation:6.3}) to get,
\begin{equation}\label{equation:6.5}
\boldsymbol{x}=\frac{A^2-B^2}{2}+\frac{A^2+B^2}{2}\cos(E)+iAB\sin(E).
\end{equation}
From equation (\ref{equation:6.5}) we see the center, which we denote $c$, must be,
\begin{equation}\label{equation:6.6}
c=\frac{A^2-B^2}{2}.
\end{equation}
Since we know $(A^2+B^2)/2\geq AB$ for all $A$ and $B$, we have,
\begin{equation}\label{equation:6.7}
a=\frac{A^2+B^2}{2},
\end{equation}
where $a$ is the semi-major axis, and,
\begin{equation}\label{equation:6.8}
b=AB,
\end{equation}
where $b$ is the semi-minor axis. We now may prove Kepler's first law.
\begin{theorem}[The Law of Orbits]\label{theorem:1}
All planets move in elliptical orbits, with the sun at one focus.
\end{theorem}
\begin{proof}
Note that the equation for the planets motion has the sun placed at the origin. Also note that the equation for motion is obviously either an ellipse, or a perfect circle. When using the semi-major axis $a$, the semi-minor axis $b$, and the center $c$, we find,
\begin{equation}\label{equation:6.9}
a^2-b^2-c^2=0.
\end{equation}
This is the definition of a focus of an ellipse, thus the planet orbits the origin. Also if the orbit is a perfect circle, the center is at zero, and thus orbits the sun. Therefore since the sun is at the origin, the planet orbits the sun at one focus.
\end{proof}

\subsection{Kepler's second law}

Since we have derived Kepler's first law, we know that
\begin{equation}\label{equation:6.10}
e=\frac{c}{a},
\end{equation}
where $e$ is the eccentricity of the orbit. We then use (\ref{equation:6.6}), (\ref{equation:6.7}), and (\ref{equation:6.10}) to solve for $A$ and $B$. We find,
\begin{equation}\label{equation:6.11}
A=\sqrt{a(1+e)}, \qquad B=\sqrt{a(1-e)}.
\end{equation}
Plugging (\ref{equation:6.11}) into (\ref{equation:6.5}) yields
\begin{equation}\label{equation:6.12}
\boldsymbol{x}=a(e+\cos(E))+ia\sqrt{1-e^2}\sin(E).
\end{equation}
Using (\ref{equation:6.12}) we calculate the magnitude of $\boldsymbol{x}$ to find,
\begin{equation}\label{equation:6.13}
r=a(1+e\cos(E)).
\end{equation}
We are now prepared to derive Kepler's second law.
\begin{theorem}[The Law of Areas]\label{theorem:3}
A line that connects a planet to the sun sweeps out equal areas in equal times.
\end{theorem}
\begin{proof}
Since $\boldsymbol{x}=a(e+\cos(E))+ia\sqrt{1-e^2}\sin(E)$, we have,
\begin{equation}\label{equation:6.14}
\boldsymbol{x}_0=a(e+\cos(E)), \qquad \boldsymbol{x}_1=a\sqrt{1-e^2}\sin(E).
\end{equation}
We then let,
\begin{equation}\label{equation:6.15}
\varphi=\tan^{-1}\left(\frac{\boldsymbol{x}_1}{\boldsymbol{x}_0}\right),
\end{equation}
and let the area swept out by the line segment be $\Sigma$.
\begin{figure}[!ht]
\centering
\begin{tikzpicture}[thick]

\coordinate (O) at (4,0);
\coordinate (A) at (0,0);
\coordinate (B) at (4,1.5);
\draw (O)--(A)--(B)--cycle;

\tkzLabelSegment[below=2pt](O,A){\textit{$r$}}
\tkzLabelSegment[right=2pt](O,B){\textit{$rd\varphi$}}
\draw[color=gray] (1,0.175) node {$\varphi$};

\end{tikzpicture}
\caption{Change in Area}\label{figure:2}
\end{figure}

From figure (\ref{figure:2}), we can define the change of area being swept out as,
\begin{equation}\label{equation:6.16}
\frac{d\Sigma}{dt}=\frac{1}{2}r^2\frac{d\varphi}{dt}=\frac{1}{2}r^2\frac{d\varphi}{dE}\frac{dE}{dt}=\frac{1}{2}r^2\frac{d\varphi}{dE}\frac{dE}{d\tau}\frac{d\tau}{dt}.
\end{equation}
We take the derivative of (\ref{equation:6.4}) to get,
\begin{equation}\label{equation:6.17}
\frac{dE}{d\tau}=\sqrt{2h}.
\end{equation}
We then substitute (\ref{equation:6.17}), and (\ref{equation:5.5}$_1$) into (\ref{equation:6.16}), which yields
\begin{equation}\label{equation:6.18}
\frac{d\Sigma}{dt}=\frac{1}{2}r^2\left(\frac{d}{dE}\tan^{-1}\frac{\boldsymbol{x}_1}{\boldsymbol{x}_0}\right)\sqrt{2h}\frac{1}{r}.
\end{equation}
We then take the derivative, cancel one $r$, substitute $r=a(1+e\cos(E))$, and move the constants to the left, which transforms (\ref{equation:6.18}) into
\begin{equation}\label{equation:6.19}
a\sqrt{\frac{h}{2}(1-e^2)}(1+e\cos(E))\left(\frac{f(E)}{g(E)}\right),
\end{equation}
where
\begin{equation}
f(E)=\frac{\cos(E)}{\cos(E)+e}+\frac{\sin^2(E)}{(\cos(E)+e)^2},
\end{equation}
and
\begin{equation}
g(E)=1-\frac{(e^2-1)\sin^2(E)}{(\cos(E)+e)^2}.
\end{equation}

\medskip
We now reduce the fraction. First we look at the numerator, $f(E)$. We multiply the top and bottom of the first fraction in the numerator by its denominator, and combine with the second fraction, yielding,
\begin{equation}\label{equation:6.22}
\frac{e\cos(E)+1}{(\cos(E)+e)^2}.
\end{equation}
We then multiply the $1$ in $g(E)$ on a common denominator, and combining with the second fraction to get,
\begin{equation}\label{equation:6.23}
\frac{(\cos(E)+e)^2-(e^2-1)\sin^2(E)}{(\cos(E)+e)^2}.
\end{equation}
Since the $f(E)$ and $g(E)$ have the same denominator, we cancel it, which reduces the fraction $f(E)/g(E)$ to
\begin{equation}\label{equation:6.24}
\frac{e\cos(E)+1}{(\cos(E)+e)^2-(e^2-1)\sin^2(E)}.
\end{equation}
Note that the numerator is identical to $r/a$. We multiply in our $r$ term to get,
\begin{equation}\label{equation:6.25}
\frac{(e\cos(E)+1)^2}{(\cos(E)+e)^2-(e^2-1)\sin^2(E)}.
\end{equation}
We then expand the numerator and denominator to get,
\begin{equation}\label{equation:6.26}
\frac{e^2\cos^2(E)+2e\cos(E)+1}{\cos^2(E)\!+\!2e\cos(E)+e^2-e^2\sin^2(E)\!+\!\sin^2(E)}.
\end{equation}
Note we can now reduce the numerator and denominator to be the same value, so
\begin{equation}\label{equation:6.27}
\frac{e^2(1-\sin^2(E))+2e\cos(E)+1}{e^2(1-\sin^2(E))+2e\cos(E)+1}=1.
\end{equation}
Thus substituting (\ref{equation:6.27}) into (\ref{equation:6.19}) yields,
\begin{equation}\label{equation:6.28}
\frac{d\Sigma}{dt}=a\sqrt{\frac{h}{2}(1-e^2)}.
\end{equation}
We now have $d\Sigma/dt$ equal to a constant, which proves Kepler's second law.
\end{proof}
\subsection{Kepler's third law}

We now work to represent the energy $h$ with variables describing the motion of the planet. We first note that (\ref{equation:5.13}) gives us one equation representing ${2\left|\boldsymbol{q}'\right|^2=2\left|\boldsymbol{u}'\right|^2}$. We then take the derivative of (\ref{equation:6.2}), and find that,
\begin{equation}\label{equation:6.29}
2\left|\boldsymbol{u}'\right|^2=ah(1-e\cos(E)).
\end{equation}
We set (\ref{equation:6.29}) and the negative of (\ref{equation:5.13}) equal, which yields $2ah=\mu$, or,
\begin{equation}\label{equation:6.30}
h=\frac{\mu}{2a}.
\end{equation}
We now are prepared to derive Kepler's third law.
\begin{theorem}[The Law of Periods]\label{theorem:2}
The square of the period of any planet is proportional to the cube of the semi-major axis of its orbit. 
\end{theorem}
\begin{proof}
Firstly we define,
\begin{equation}\label{equation:6.31}
n:=\sqrt{\frac{\mu}{a^3}}.
\end{equation}
We square and divide (\ref{equation:6.31}) by two to show
\begin{equation}\label{equation:6.32}
\frac{n^2}{2}=\frac{\mu}{2a^3}.
\end{equation}
We substitute (\ref{equation:6.30}) into (\ref{equation:6.32}) to get
\begin{equation}\label{equation:6.33}
\frac{n^2}{2}=\frac{h}{a^2}.
\end{equation}
We then multiply by two, take the square root, and substitute (\ref{equation:6.17}) to get
\begin{equation}\label{equation:6.34}
n=\frac{1}{a}\frac{dE}{d\tau}=\frac{1}{a}\frac{dE}{dt}\frac{dt}{d\tau}.
\end{equation}
\\From (\ref{equation:5.5}$_1$) we substitute $dt/d\tau$ and rearrange to get
\begin{equation}\label{equation:6.35}
ndt=\frac{r}{a}dE.
\end{equation}
We then substitute (\ref{equation:6.13}) in and take the definite integral over the period $T$, in which the eccentric anomaly goes from $0$ to $2\pi$,
\begin{eqnarray}
\int_0^T ndt&=&\int_0^{2\pi}\frac{a(1+e\cos(E))}{a}dE\label{equation:6.36},\\
nT&=&2\pi\label{equation:6.37}.
\end{eqnarray}
Finally we substitute (\ref{equation:6.31}) into (\ref{equation:6.37}), square both sides, and rearrange to yield
\begin{equation}\label{equation:6.38}
\frac{T^2}{a^3}=\frac{4\pi^2}{\mu}.
\end{equation}
Since we have the ratio of the period squared to the semi-major axis cubed being a constant, the derivation is complete.
\end{proof}
Now that Kepler's third law is derived, we note that $n$ is actually the \textit{mean motion} of the planet, and rearranging (\ref{equation:6.37}) shows that,
\begin{equation}\label{equation:6.39}
n=\frac{2\pi}{T}.
\end{equation}
We summarize by deriving the area of an ellipse.
\begin{theorem}[Area of an Ellipse]
The area of the ellipse generated by the orbit of the planet, $A_e$, is found to be
\begin{equation}\label{equation:6.40}
A_e=\pi a b.
\end{equation}
\end{theorem}
\begin{proof}
We substitute (\ref{equation:6.11}) into (\ref{equation:6.8}) to find
\begin{equation}\label{equation:6.41}
b=a\sqrt{1-e^2}.
\end{equation}
We then substitute (\ref{equation:6.41}) into (\ref{equation:6.28}) to find
\begin{equation}\label{equation:6.42}
\frac{d\Sigma}{dt}=b\sqrt{\frac{h}{2}}.
\end{equation}
We then multiply the $dt$ on both sides, and integrate. We want to integrate over the whole area of the ellipse, $A_e$, which takes the period, $T$, time, so we have
\begin{equation}\label{equation:6.43}
\begin{split}
\int_0^{A_e} d\Sigma=&\int_0^Tb\sqrt{\frac{h}{2}}dt\\
A_e=&bT\sqrt{\frac{h}{2}}.
\end{split}
\end{equation}
We then substitute (\ref{equation:6.30}), and (\ref{equation:6.37}) into (\ref{equation:6.43}),
\begin{equation}\label{equation:6.44}
A_e=b\frac{2\pi}{n}\sqrt{\frac{\mu}{4a}}.
\end{equation}
Thus after substituting in (\ref{equation:6.31}) and canceling terms, (\ref{equation:6.44}) yields
\begin{equation}\label{equation:6.45}
A_e=\pi a b.
\end{equation}
\end{proof}

\bibliography{main}
\end{document}